\documentclass[1p]{elsarticle}
\usepackage{amsmath,amsfonts,amssymb,amsthm,times,mathptmx,graphicx}
\usepackage{epstopdf}
\usepackage[arrow, matrix, curve]{xy}
\usepackage{tikz}
\newtheorem{theorem}{Theorem}
\newtheorem{lemma}{Lemma}
\newtheorem{definition}{Definition}

\newtheorem{remark}{Remark}
\theoremstyle{remark}
\newtheorem{example}{Example}
\usepackage[utf8]{inputenc}

\newcommand\set[1]{\left\{#1\right\}}

\newcommand\norm[1]{\left\lVert#1\right\rVert}
\newcommand\map[3]{#1:#2\to#3}

\newcommand\NRp{\mathbb{R}_{\ge0}}

\begin{document} 
\title{On the Balance of Unrooted Trees}
\author{Mareike Fischer}
\ead{email@mareikefischer.de}
\author{Volkmar Liebscher}
\ead{volkmar.liebscher@uni-greifswald.de}
\address{Ernst-Moritz-Arndt University Greifswald, Institute for  Mathematics and Computer Science, Walther-Rathenau-Stra{\ss}e 47
D-17487 Greifswald/Germany} 
\date{\day27\month10\year2015\relax\today
}

\begin{abstract}
We solve a class of optimization problems for (phylogenetic) $X$-trees or their shapes. These problems have recently appeared in different contexts, e.g. in the context of the impact of tree shapes on the size of TBR neighborhoods, but so far these problems have not been characterized and solved in a systematic way. In this work we generalize the concept and also present several applications. Moreover, our results give rise to a nice notion of balance for trees. Unsurprisingly, so-called caterpillars are the most unbalanced tree shapes, but it turns out that balanced tree shapes cannot be described so easily as they need not even be unique. \end{abstract}
\maketitle
\section{Introduction}

When phylogenetic trees are considered, i.e. trees describing the evolutionary history of $n$ present-day species which label the leaves of the tree, one is often confronted with the need to find the extreme values of the expression 

 \begin{displaymath}
  \Phi_f(\tau)= \sum _{\sigma\in \Sigma^*(\tau)}f(\norm\sigma).
 \end{displaymath}
 
Here, $\tau$ varies over all phylogenetic trees with $n$ leaves, and $\Sigma^*(\tau)$ is the set of so-called non-trivial splits of $X=\set{1,\dots,n}$ induced by $\tau$. Recall that a split $\sigma$ of a set $X$ is just a bipartition of this set into two non-empty subsets $A$ and $B$ (we then write $\sigma=A|B$), and a split is called non-trivial whenever both parts of the bipartition have cardinality at least 2. Splits of the species set $X$ play an important role in mathematical phylogenetics, because every edge of a phylogenetic tree induces such a split, and splits are non-trivial if and only if they are induced by an inner edge of the tree (i.e. not by an edge connected to a leaf). Moreover, in the above definition of $ \Phi_f(\tau)$, $\norm{A|B}=\min(|A|,|B|)$ denotes the cardinality of the smaller part of the split. 

Such expressions recently appeared in different contexts, each time with a different choice of the monotone function $\map f{\set{2,\dots,\lfloor n/2\rfloor}}\NRp$:
 \begin{enumerate}
 \item The size of a TBR-neighborhood of the tree $\tau$ with
   $f(k)=k(n-k)$, see \cite{HW12}. This example is also related to the so-called Wiener index of the tree \cite{SWW11}. For this index, you need to assign edge lengths to all edges of $\tau$. Then, the Wiener index is defined as $W(\tau) = \sum\limits_{u,v \in V(T)} d_\tau(u,v)$, where $V(\tau)$ refers to the vertex set of $\tau$ and $d_\tau$ refers to the pairwise distances between any two vertices induced by the edge lengths.   
 \item An estimate of the diameter of the (unweighted) tree space under Gromov-type distance measures introduced in \cite{Lie15} with $f(k)=k$ ($\ell^1-$Gromov) and $f(k)=\sqrt{k(n-k)}$ ($\ell^2-$Gromov).
 \item The number of so-called cherries of a tree, i.e. splits $\sigma=A|B$ with $\norm{A|B}=\min(|A|,|B|)=2$, can be described with $f(k)=\left\{
     \begin{array}[c]{cl}
       1&k=2\\0&k>2
     \end{array}
\right.$.
 \end{enumerate}
 
In many situations, it turns out that the extremal shapes (giving the minimal or maximal values of the functional) are so-called caterpillars (i.e. rooted trees with just one cherry or unrooted trees with just two cherries) and so-called semi-regular trees, i.e. trees which have at most one non-leaf vertex that is not of maximum degree and that such a vertex, if it exists, cannot be adjacent to more than one non-leaf vertex \cite{SWW11}. While the first kind of trees, the caterpillars, are often considered the most unbalanced tree shapes, the latter kind, i.e. the semi-regular trees, are considered to be most balanced. This reoccurence is challenging, as there seems to be no study of the whole family of functionals with $f$ broadly varying. We regard the present note as the beginning of this study. 

But how is $\Phi_f(\tau)$ measuring the balance of the shape of $\tau$? Intuitively, splits $\sigma\in\Sigma^*(\tau)$ with a high value of $\norm\sigma$ are very balanced, and it seems that the maximum of $\Phi_f(\tau)$ should be attained by the most balanced shapes. But, more importantly, the balanced tree shapes display  \emph{more}  splits $\sigma$ with a small $\norm\sigma$, for example cherries. Thus they realise a smaller  value of $\Phi_f$ than caterpillars.
  
In fact, the Arxiv version of \cite{HW13}, namely \cite{HW12}, provided already the main idea for deriving the maximal value for increasing functions $f$. As this Arxiv version provides a few more details than the published manuscript, we will subsequently refer to that version. Anyway, a lot of structure is required for deriving the minimal value, whereas the maximum is significantly easier to prove, cf. \cite{HW12,SWW11}. 

We will show in the following that for general functions $f$, the main principle becomes even more transparent when considering particular functionals. Last but not least, we also give further applications to topological indices. 

\section{Preliminaries}

Let $\mathcal T_n$ be the set of all (unrooted) phylogenetic trees (i.e. connected acyclic graphs with leaves labelled by a so-called taxon set $X$) with $|X|=n\ge 6$ leaves,  and $\mathcal T^2_n$ the subset of $\mathcal T_n$ which contains all fully resolved (i.e. binary, bifurcating) trees in $\mathcal T_n$; i.e. the trees in $\mathcal T^2_n$ have the property that all vertices have either degree 3 (inner nodes) or 1 (leaves). When there is no ambiguity, we often just say tree when referring to a phylogenetic tree or, when the leaf labeling is not important, its so-called tree shape, respectively. 

Caterpillars with $n$ leaves are binary phylogenetic trees with two leaves, say $1,n$, such that every vertex of the tree is on the path from $1$ to $n$ or adjacent to a vertex on this path. In the so-called Newick format \cite{newick}, which gives a nested list of all leaves such that leaves which are separated by fewer edges in the tree are also separated by fewer brackets, caterpillars may be denoted by the expression
 \begin{displaymath}
   \tau_c=(((1,2),3),\ldots, n).
 \end{displaymath}
If $\tau\in \mathcal T_n$, let $\Sigma(\tau)$ denote the set of all splits, i.e. all partitions of the leaf set $X$ into two non-empty subsets $A$ and $B$, and let $\Sigma^*(\tau)$ be the set of all non-trivial splits $\sigma=A|B$ induced by inner edges of $\tau$, i.e. $\Sigma^*(\tau)$ contains all splits for which both $|A| \geq 2$ and $|B|\geq 2$.  For a split $\sigma = A|B$ let  $\norm\sigma=\norm{A|B}=\min(|A|,|B|)$ denote its size.

Now we introduce for any function $\map f{\set{2,\dots,\lfloor n/2\rfloor}}\NRp$  the functional $\map{\Phi_f}{\mathcal T_n}\NRp$ via
 \begin{displaymath}
   \Phi_f(\tau)=\sum _{\sigma\in \Sigma^*(\tau)}f(\norm\sigma).
 \end{displaymath}

Clearly, a tree and its contraction, obtained by suppressing all inner vertices of degree 2, share the same value of $\Phi_f$. A rotation of the tree, which is  obtained by permuting the leaf labels, does not alter the value of $\Phi_f$ either. So $\Phi_f$ is just a function of the phylogenetic tree shape of $\tau$. Particularly, all caterpillars get  the same value under $\Phi_f$.     

The last concept we need to introduce before we can present our results are so-called NNI-moves on binary trees. NNI stands for Nearest Neighbor Interchange, and in order to perform an NNI-move on a binary tree $\tau$, you fix an inner edge of $\tau$. This edge is connected to four subtrees, two on either side. You then swap two of these subtrees from opposite sides of the edge. This procedure is called NNI-move, and the resulting tree $\tau'$ is called an NNI-neighbor of $\tau$. You can also define a metric based on NNI in order to measure the distance between two trees $\tau$, $\tau'$ such that $d(\tau,\tau')$ equals the minimum number of NNI moves needed to get from $\tau$ to $\tau'$.

\section{Results for increasing functions $f$}

In this section, we consider a function $\map f{\set{1,\dots,\lfloor n/2\rfloor}}\NRp$ with the additional assumption that $f$ is monotonously increasing, i.e. if $x >y$, then $f(x) \geq f(y)$. The following theorem, which is based on ideas presented in \cite{HW12}, shows that caterpillars maximize $\Phi_f$ in this case.

\begin{theorem}\label{thm:1}
  Let $\tau_c\in\mathcal T_n$ be a caterpillar and $\map f{\set{2,\dots,\lfloor n/2\rfloor}}\NRp$ monotonously increasing. Then for all $\tau\in\mathcal T_n$
  \begin{displaymath}
    \Phi_f(\tau_c)\ge\Phi_f(\tau).
  \end{displaymath}

If $f$ is strictly increasing and $\tau\in \mathcal T_n$ is a point of maximum of $\Phi_f$, then $\tau$ is a caterpillar. 
\end{theorem}
Before presenting the proof, we consider the following elementary lemma, which we will need subsequently. 

\begin{lemma}\label{lem:turmsum}
Let  $\map f{\set{2,\dots,\lfloor n/2\rfloor}}\NRp$ be strictly monotonously increasing. Then for all $m\in \set{2,\dots,n-2}$ and $p\in \set{1,\dots,n-3-m}$
\begin{displaymath}
  \sum_{l=2}^mf(\min(l,n-l))<  \sum_{l=2+p}^{m+p}f(\min(l,n-l)). 
\end{displaymath}
\end{lemma}
\begin{proof}
  For symmetry reasons we may assume that $p\le \lfloor(n-2-m)/2\rfloor$ and consider the function $\map g{\set{0,\dots,\lfloor(n-2-m)/2\rfloor}}\NRp$,
  \begin{displaymath}
g(p)=\sum_{l=2+p}^{m+p}f(\min(l,n-l)).
\end{displaymath}
We find
  \begin{displaymath}
    g(p+1)-g(p)=f(\min(m+p+1,n-m-p-1))-f(\min(2+p,n-2-p))>0
  \end{displaymath} 
since $\min(m+p+1,n-m-p-1)>\min(2+p,n-2-p)$.
\end{proof}
\begin{proof}[Proof of Theorem \ref{thm:1}.]
Consider first the case that $f$ is strictly increasing and $\tau$ is such that $\Phi_f(\tau)$ is maximal. 
We now follow directly the arguments in the proof of \cite[Lemma 4.1]{HW12}. 

Fix two cherries $x_1,x_2$ and $x_3,x_4$ of $\tau$  and partition $\set{1,\dots,n}$ by the edges of the path from $x_1$ to $x_4$, see Fig. \ref{fig:proof1}.
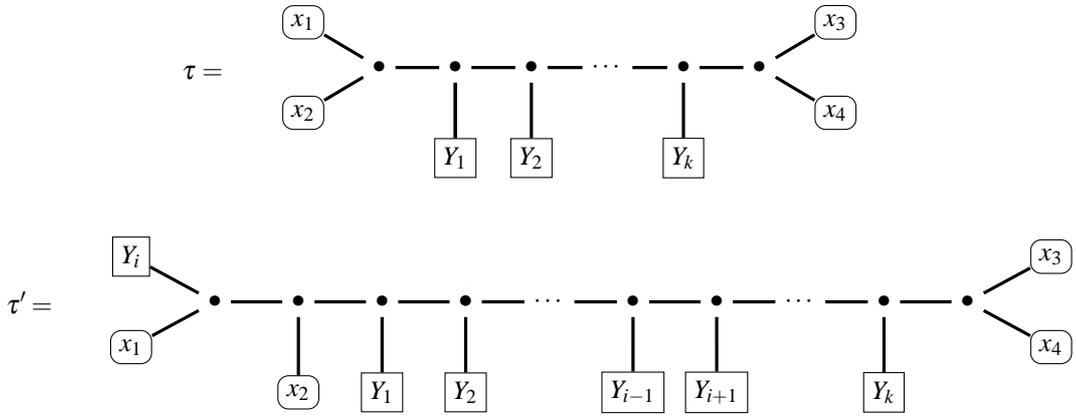
\begin{figure}[htbp]
  \centering
  \begin{displaymath}
\tau=\qquad\begin{tikzpicture}[xscale=1,yscale=0.6, baseline =30 ]
\node (x1) at (0,3) [rectangle,rounded corners,draw=black] {$x_1$};
\node (x2) at (0,1) [rectangle,rounded corners,draw=black] {$x_2$};
\node (z1) at (1,2)  {$\bullet$};
\node (z2) at (2,2)  {$\bullet$};
\node (z3) at (3,2)  {$\bullet$};
\node (ppp) at (4,2)  {$\cdots$};
\node (zk) at (5,2)  {$\bullet$};
\node (zk+1) at (6,2)  {$\bullet$};
\node (x3) at (7,3) [rectangle,rounded corners,draw=black] {$x_3$};
\node (x4) at (7,1) [rectangle,rounded corners,draw=black] {$x_4$};
\node (y1) at (2,0) [rectangle,draw=black] {$Y_1$};
\node (y2) at (3,0) [rectangle,draw=black] {$Y_2$};
\node (yk) at (5,0) [rectangle,draw=black] {$Y_k$};
\draw[very thick] (x1) to (z1);
\draw[very thick] (x2) to (z1);
\draw[very thick] (y1) to (z2);
\draw[very thick] (y1) to (z2);
\draw[very thick] (y2) to (z3);
\draw[very thick] (yk) to (zk);
\draw[very thick] (z1) to (z2);
\draw[very thick] (z2) to (z3);
\draw[very thick] (z3) to (ppp);
\draw[very thick] (ppp) to (zk);
\draw[very thick] (zk) to (zk+1);
\draw[very thick] (x3) to (zk+1);
\draw[very thick] (x4) to (zk+1);
\end{tikzpicture}
\end{displaymath}

\begin{displaymath}
\tau'=\qquad\begin{tikzpicture}[xscale=1.1,yscale=0.6, baseline =30 ]
\node (yi) at (-1,3) [rectangle,draw=black] {$Y_i$};
\node (x1) at (-1,1) [rectangle,rounded corners,draw=black] {$x_1$};
\node (x2) at (1,0) [rectangle,rounded corners,draw=black] {$x_2$};
\node (z0) at (0,2)  {$\bullet$};
\node (z1) at (1,2)  {$\bullet$};
\node (z2) at (2,2)  {$\bullet$};
\node (z3) at (3,2)  {$\bullet$};
\node (z4) at (3,2)  {$\bullet$};
\node (ppp1) at (4,2)  {$\cdots$};
\node (zi) at (5,2)  {$\bullet$};
\node (zi+1) at (6,2)  {$\bullet$};
\node (ppp2) at (7,2)  {$\cdots$};
\node (zk) at (8,2)  {$\bullet$};
\node (zk+1) at (9,2)  {$\bullet$};
\node (x3) at (10,3) [rectangle,rounded corners,draw=black] {$x_3$};
\node (x4) at (10,1) [rectangle,rounded corners,draw=black] {$x_4$};
\node (y1) at (2,0) [rectangle,draw=black] {$Y_1$};
\node (y2) at (3,0) [rectangle,draw=black] {$Y_2$};
\node (yi-1) at (5,0) [rectangle,draw=black] {$Y_{i-1}$};
\node (yi+1) at (6,0) [rectangle,draw=black] {$Y_{i+1}$};
\node (yk) at (8,0) [rectangle,draw=black] {$Y_k$};
\draw[very thick] (x1) to (z0);
\draw[very thick] (yi) to (z0);
\draw[very thick] (x2) to (z1);
\draw[very thick] (z0) to (z1);
\draw[very thick] (y1) to (z2);
\draw[very thick] (y2) to (z3);
\draw[very thick] (yi-1) to (zi);
\draw[very thick] (yi+1) to (zi+1);
\draw[very thick] (z1) to (z2);
\draw[very thick] (z2) to (z3);
\draw[very thick] (z3) to (z4);
\draw[very thick] (z4) to (ppp1);
\draw[very thick] (ppp1) to (zi);
\draw[very thick] (zi) to (zi+1);
\draw[very thick] (zi+1) to (ppp2);
\draw[very thick] (ppp2) to (zk);
\draw[very thick] (yk) to (zk);
\draw[very thick] (zk) to (zk+1);
\draw[very thick] (x3) to (zk+1);
\draw[very thick] (x4) to (zk+1);
\end{tikzpicture}
\end{displaymath}

\caption{The structure of the trees $\tau$ and $\tau'$ in the proof of Theorem \ref{thm:1}. The labels with round edges refer to leaves $x_1, \ldots,x_4$, whereas the squares refer to subtrees $Y_1,\ldots, Y_k$.}
\label{fig:proof1}
\end{figure}

We just want to prove that the number of leaves in $Y_i$, which we denote by $|Y_i|=y_i$, equals 1 for all $i=1,\dots,k$, because this is equivalent to $\tau$ being a caterpillar. So let us assume there is an $i$ with $y_i\ge 2$. We choose $i$ minimal, i.e. $y_1=y_2=\cdots=y_{i-1}=1$. 
We now construct a new tree $\tau'$ by moving $Y_i$ such that it is now a neighbor of $x_1$, see again Figure \ref{fig:proof1}.   We now consider several kinds of non-trivial splits in $\tau'$:
\begin{enumerate}
\item Splits $\sigma$ which are not on the path from $x_1$ to $x_4$. For those splits $\norm \sigma$ is the same for $\tau'$ as for $\tau$.
\item Splits $\sigma$ on the path from $x_1$ to $x_4$ before the edge splitting $Y_{i-1}$ from $Y_{i+1}$ (i.e. in Figure \ref{fig:proof1} to the left of this split). Enumerating these splits by their distance from $x_1$ in $\tau$ and $\tau'$,  the size of the $l^{th}$ split changes from $\min(l+1,n-l-1)$  to $\min(l+y_i,n-l-y_i)$. 
\item The edge on the path from $x_1$ to $x_4$ which separates $Y_{i-1}$ from $Y_{i+1}$. Here, the size of $\sigma$ in $\tau'$ equals the size of the split splitting $Y_i$ from $Y_{i+1}$ in $\tau$. 
 \item Splits $\sigma$ on the path from $x_1$ to $x_4$ after the edge splitting $Y_{i-1}$ from $Y_{i+1}$, i.e. in Figure \ref{fig:proof1} at the right-hand side. Here, the size of $\sigma$ remains unchanged from $\tau$ to $\tau'$.
\end{enumerate}
Thus we find 
\begin{displaymath}
  \Phi_f(\tau')-\Phi_f(\tau)=\sum_{l=1}^if(\min(l+y_i,n-l-y_i))-f(\min(l+1,n-l-1))>0
\end{displaymath}
by Lemma \ref{lem:turmsum}.  However, this implies that  $\tau$ is not optimal. This is a contradiction, because $\tau$ was chosen to be optimal. So we conclude that the assumption is wrong and $y_i=1$ for all $i=1, \ldots, k$. Therefore, $\tau$ is a caterpillar. 

If $f$ is increasing but not strictly increasing take an arbitrary tree $\tau$ and a caterpillar $\tau_c$. For an arbitrary $\varepsilon>0$ consider the function $\map {f_\varepsilon}{\set{1,\dots,\lfloor n/2\rfloor}}\NRp$,
\begin{displaymath}
f_\varepsilon(k)=f(k)+\varepsilon k \mbox{ for all } k=1,\dots,\lfloor\frac n2\rfloor.
\end{displaymath}
$f_\varepsilon$ is strictly increasing and we obtain from the previous arguments 
\begin{displaymath}
  \Phi_{f_\varepsilon}(\tau_c)\ge\Phi_{f_\varepsilon}(\tau).
\end{displaymath}
For $\varepsilon\downarrow0$  we have $ \Phi_{f_\varepsilon}(\tau_c)\downarrow \Phi_{f}(\tau_c)$ and  $ \Phi_{f_\varepsilon}(\tau)\downarrow \Phi_{f}(\tau)$, which completes the proof.
\end{proof}

Looking for an unconstrained minimum over $\mathcal T_n$ does not make much sense since it is attained at the most unresolved tree $\tau_0$ with no inner splits: $\Sigma^*(\tau)=\emptyset$. Note that such a tree is also often referred to as a star-tree. If we restrict our consideration to $\mathcal{T}^2_n$, instead of considering the minimum of an increasing function, we could  equivalently consider decreasing functions $f$ and look again at the maximum. Anyway, the analysis is much more difficult, which is why we first need to present an important concept needed in this context, namely the so-called split size sequences.

\section{Results for split size sequences and minima}

From now on, we focus on binary trees, and we now introduce so-called split size sequences. In order to obtain these, we first associate to every binary tree $\tau \in \mathcal T^2_n$ the (multi-)set of split sizes $\set{\norm \sigma:\sigma\in \Sigma^*(\tau)}$. However, for ease of notation, it is better to work with  ordered $n-3$ tuples instead of (unordered) sets, which is why we continue with the following definition.

  \begin{definition}
 Let  $\tau \in \mathcal T^2_n$. Then, $\tau$ induces $n-3$ non-trivial splits, i.e. $|\Sigma^*(\tau)|=n-3$. We assume an arbitrary ordering $\sigma_1,\ldots, \sigma_{n-3}$ of these splits and define the $(n-3)$-tuple $\tilde{s}(\tau)$ as follows: 
 
 $$ \tilde{s}(\tau)_i = \norm{\sigma_i} \mbox{ for all } i=1,\ldots,n-3.$$
 We then define the split size sequence $s(\tau)$ as follows: Order the $n-3$ entries of $\tilde{s}(\tau)$ increasingly and call the resulting ordered sequence $s(\tau)$. This is the split size sequence. Moreover, we denote by $\mathcal S_n=\set{s(\tau):\tau \in \mathcal T_n^2}$ the set of all split size sequences on $n$ taxa.
 \end{definition}
 
As an example for the split size sequence, we now consider the caterpillar tree.  Consider again Figure \ref{fig:proof1}. If $\tau$ in this figure is a caterpillar, denoted $\tau_c$, it has $|Y_1|=\dots=|Y_k|=1$. In order to get an (unordered) sequence of all split sizes, we start at one cherry, say $x_1,x_2$ and subsequently consider the splits $\{x_1,x_2,Y_1\}|X\setminus \{x_1,x_2,Y_1\}, \ldots, \{x_1,x_2,Y_1,\ldots,Y_k\}|X\setminus \{x_1,x_2,Y_1,\ldots,Y_k\}$. It is clear that the cherry contributes the value 2 to $\tilde{s}(\tau_c)$, then we get a 3, 4, 5, and so forth. However, as soon as there are fewer leaves on the right-hand side, say in set $B$, than on the left-hand side, say in set $A$, where $\sigma=A|B$, the sequence will continue with $|B|=n-|A|$ instead of $|A|$. In the end, we order the elements of $\tilde{s}(\tau_c)$ increasingly in order to derive $s(\tau_c)$. For example, if the caterpillar has $n=6$ leaves, we get $s(\tau_c)=(2,2,3)$, if $n=7$, we get $s(\tau_c)=(2,2,3,3)$, if $n=8$, we get $s(\tau_c)=(2,2,3,3,4)$, and so forth.

Note that there is an alternative equivalent definition of $s(\tau)$. At first, it  might seem a little less intuitive, but it proves to be useful in the following. Consider the increasing sequence $m(\tau)\in \mathbb{N}^{\set{2,\ldots,\lfloor n/2\rfloor}}$ whose entries are defined as follows: 
   \begin{displaymath}
     m(\tau)_j=\left\lvert\set{\sigma\in \Sigma^*(\tau):\norm \sigma \le j}\right\rvert \quad\mbox{ for all } j=2,\ldots, \lfloor {n}/{2} \rfloor. 
   \end{displaymath}
Then, $s(\tau)\in \mathbb{N}^{n-3}$ is just the left inverse of $m(\tau)$: \begin{displaymath}
  s(\tau)_i=\min\set{j:m(\tau)_j\ge i}.
\end{displaymath}
This means 
\begin{displaymath}
     s(\tau)_i=j\qquad\mbox{~if~} m(\tau)_j\ge i\mbox{~and~} m(\tau)_{j-1} < i.
   \end{displaymath}

The function $m$ will be used in the subsequent proofs.


\par \vspace{0.5cm}Note that it is not so clear how to characterize $\mathcal S_n$, i.e. all possible split size sequences. 
{However, in the following we investigate them a bit further.} Therefore, first of all we  order $\mathcal S_n$ in the pointwise sense. 
So we say that $\tau\in \mathcal T_n^2$ is more balanced than $\tau'\in\mathcal T_n^2$, denoted $\tau\prec\tau'$, if
\begin{displaymath}
  s(\tau)_i\le s(\tau')_i\qquad \mbox{ for all }i=1,\dots,n-3,
  \end{displaymath} or equivalently, 
\begin{displaymath}
  m(\tau)_j\ge m(\tau')_j\qquad \mbox{ for all }j=2,\dots,\lfloor  n/2\rfloor.
\end{displaymath}

This leads to the following theorem.

\begin{theorem} \label{thm:2}Let $\tau, \tau' \in \mathcal T_n^2$. Then, we have $\tau\prec\tau'$ if and only if for all monotonously increasing  $\map f{\set{2,\dots,\lfloor n/2\rfloor}}\NRp$
  \begin{equation}\label{eq:dominance}
    \Phi_f(\tau)\le\Phi_f(\tau').
  \end{equation}
\end{theorem}
\begin{proof}
  By definition, we have 
$$
    \Phi_f(\tau)=\sum_{i=1}^{n-3}f(s(\tau)_i) \leq 
   \sum_{i=1}^{n-3}f(s(\tau')_i)  = \Phi_f(\tau'),
$$
where the inequality is due to $\tau\prec\tau'$, which implies $s(\tau)_i \leq s(\tau')_i$ for all $i=1,\ldots,n-3$, and $f$ being monotonously increasing. This completes the first part of the proof.

Now assume we have (\ref{eq:dominance}) for all  monotonously increasing  $\map f{\set{2,\dots,\lfloor n/2\rfloor}}\NRp$. Then, in particular (\ref{eq:dominance}) holds for the following family of monotonously increasing functions: 
For each $j\in \set{2,\dots,\lfloor n/2\rfloor}$, define $f_j(k)=\left\{
  \begin{array}[c]{cl}
    1&k\ge j\\
0&k<j
  \end{array}.
\right.$
We obtain 
$$ \sum_{i=1}^{n-3}f_j(s(\tau)_i) =\Phi_{f_j}(\tau)\leq \Phi_{f_j}(\tau')=
   \sum_{i=1}^{n-3}f_j(s(\tau')_i)\quad   \mbox{ for all } j = 2,\ldots, \lfloor n/2\rfloor. $$

By the definition of $f_j$, this implies 
\begin{displaymath}
\left\lvert\set{\sigma\in \Sigma^*(\tau):\norm \sigma \ge j}\right\rvert\ge \left\lvert\set{\sigma\in \Sigma^*(\tau'):\norm \sigma \ge j}\right\rvert.
\end{displaymath}

This immediately leads to  $m(\tau)_j\ge m(\tau')_j$ for all $j=2,\dots,\lfloor {n}/2\rfloor$  and thus $\tau\prec\tau'$. 
\end{proof}

\begin{remark}
From an abstract point of view the result is almost obvious: The functions $f_j(k)=\left\{
  \begin{array}[c]{cl}
    0&k<j\\
1&k\ge j
  \end{array}
\right.$, which we used in the proof of the above theorem, generate the extremal rays of the convex cone $\set{\map f{\set{2,\dots,\lfloor n/2\rfloor}}\NRp:f\nearrow}$ and thus determine the order $\prec$. On the other hand, $\Phi_{f_j}(\tau)=n-3-m_{j-1}(\tau)$. 
\end{remark}
\begin{theorem}
  The only maximal point of $\mathcal S_n$ is derived from the sequence $m(\tau)_j=\min(2j-2,n-3)$ or 
  $s(\tau)_i=\lfloor(i+3)/2\rfloor$ corresponding to caterpillar trees $\tau$.  
\end{theorem}
\begin{proof}
  That caterpillars give the only maximal point on $\mathcal S_n$  is derived easily from the previous two theorems, as for a caterpillar $\tau_c$ we have  $\Phi_f(\tau_c)\geq \Phi_f(\tau)$ for all $\tau \in \mathcal{T}_n^2$ by Theorem \ref{thm:1}, and thus, by Theorem \ref{thm:2} we conclude $\tau\prec\tau_c$ for all $\tau$. The split size sequence of caterpillars has already been described above.
\end{proof}

Thus, the Pareto maximum of $\mathcal S_n$ is unique.  It even corresponds to  a unique tree shape.
 Unfortunately, $s(\tau)$ does not determine the shape of $\tau$ in general.  This means $\prec$ does not induce a partial order on tree shapes. We illustrate this with the following example. 
 
 \begin{example}
   The basis for this example are the following two trees: 
\begin{displaymath}
\tau=  \begin{tikzpicture}[xscale=0.5,yscale=0.2, baseline =10 ]
\node (A) at (0,4) [rectangle,draw=black] {$T_1$};
\node (B) at (0,0) [rectangle,draw=black] {$T_1$};
\node (C) at (2,2)  {$\bullet$};
\node (D) at (4,2)  {$\bullet$};
\node (E) at (6,0) [rectangle,draw=black] {$T_2$};
\node (F) at (6,4) [rectangle,draw=black] {$T_2$};
\draw[very thick] (A) to (C);
\draw[very thick] (B) to (C);
\draw[very thick] (C) to (D);
\draw[very thick] (D) to (E);
\draw[very thick] (D) to (F);
\end{tikzpicture}\qquad\mbox{ and }\qquad \tau'=\begin{tikzpicture}[xscale=0.5,yscale=0.2, baseline =10 ]
\node (A) at (0,4) [rectangle,draw=black] {$T_1$};
\node (B) at (0,0) [rectangle,draw=black] {$T_2$};
\node (C) at (2,2)  {$\bullet$};
\node (D) at (4,2)  {$\bullet$};
\node (E) at (6,0) [rectangle,draw=black] {$T_1$};
\node (F) at (6,4) [rectangle,draw=black] {$T_2$};
\draw[very thick] (A) to (C);
\draw[very thick] (B) to (C);
\draw[very thick] (C) to (D);
\draw[very thick] (D) to (E);
\draw[very thick] (D) to (F);
\end{tikzpicture}
\end{displaymath}
for two (rooted) subtrees $T_1$, $T_2$ on the same number of leaves with different shape. $\tau$ and $\tau'$ are just one NNI-move apart, which means that only two subtrees have to swap their position across one inner edge. Note that $\tau$ and $\tau'$ show the same split size sequences: First of all, both trees contain the splits induced by edges inside $T_1$ and $T_2$, as well as the four splits separating either of the two copies of $T_1$ or either of the two copies of $T_2$ from the rest of the tree. However, $\tau$ and $\tau'$ are different because they differ in the split induced by the central edge: this split separates the two copies of $T_1$ from the two copies of $T_2$ in $\tau$ and thus differs from the edge separating one copy of both $T_1$ and $T_2$ from another copy of $T_1$ and $T_2$. But even for these two splits the split size is the same, namely $n/2$, because $T_1$ and $T_2$ have the same number of leaves. Clearly, $\tau$ and $\tau'$ differ in shape if $T_1,T_2$ do so. The simplest situation is for $|T_1|=|T_2|=4$,
 \begin{displaymath}
T_1=  \begin{tikzpicture}[xscale=0.5,yscale=0.8, baseline =30 ]
\node (A) at (0,0) [rectangle,rounded corners,draw=black] {$1$};
\node (B) at (2,0) [rectangle,rounded corners,draw=black] {$2$};
\node (C) at (1,1)  {$\bullet$};
\node (D) at (2,2)  {$\bullet$};
\node (E) at (3,3)  {$\bullet$};
\node (F) at (4,0) [rectangle,rounded corners,draw=black] {$3$};
\node (G) at (6,0) [rectangle,rounded corners,draw=black] {$4$};
\draw[very thick] (A) to (C);
\draw[very thick] (B) to (C);
\draw[very thick] (C) to (D);
\draw[very thick] (D) to (E);
\draw[very thick] (F) to (D);
\draw[very thick] (G) to (E);
\end{tikzpicture}\qquad\mbox{ and }\qquad T_2=\begin{tikzpicture}[xscale=0.5,yscale=0.8, baseline =30 ]
\node (A) at (0,0) [rectangle,rounded corners,draw=black] {$1$};
\node (B) at (2,0) [rectangle,rounded corners,draw=black] {$2$};
\node (C) at (1,1)  {$\bullet$};
\node (D) at (3,3)  {$\bullet$};
\node (E) at (5,1)  {$\bullet$};
\node (F) at (4,0) [rectangle,rounded corners,draw=black] {$3$};
\node (G) at (6,0) [rectangle,rounded corners,draw=black] {$4$};
\draw[very thick] (A) to (C);
\draw[very thick] (B) to (C);
\draw[very thick] (C) to (D);
\draw[very thick] (D) to (E);
\draw[very thick] (F) to (E);
\draw[very thick] (G) to (E);
\end{tikzpicture}.\end{displaymath}  

Thus we obtain
\begin{displaymath}
  \tau=\begin{tikzpicture}[xscale=1,yscale=1,baseline=0]
\node (A14b) at (-4,4) [rectangle,rounded corners,draw=black] {$4$};
\node (A13b) at (-2,4) [rectangle,rounded corners,draw=black] {$3$};
\node (A12b) at (-1,3) [rectangle,rounded corners,draw=black] {$2$};
\node (A1b) at (0,2) [rectangle,rounded corners,draw=black] {$1$};
\node (A) at (0,0)  {$\bullet$};
\node (A1) at (-1,1)  {$\bullet$};
\node (A2) at (-1,-1)  {$\bullet$};
\node (A12) at (-2,2)  {$\bullet$};
\node (A13) at (-3,3)  {$\bullet$};
\node (B) at (2,0)  {$\bullet$};
\draw[very thick] (A) to (B);
\draw[very thick] (A) to (A1);
\draw[very thick] (A) to (A2);
\draw[very thick] (A1) to (A1b);
\draw[very thick] (A1) to (A12);
\draw[very thick] (A12) to (A12b);
\draw[very thick] (A12) to (A13);
\draw[very thick] (A13) to (A13b);
\draw[very thick] (A13) to (A14b);
\node (A24b) at (-4,-4) [rectangle,rounded corners,draw=black] {$8$};
\node (A23b) at (-2,-4) [rectangle,rounded corners,draw=black] {$7$};
\node (A22b) at (-1,-3) [rectangle,rounded corners,draw=black] {$6$};
\node (A21b) at (0,-2) [rectangle,rounded corners,draw=black] {$5$};
\node (A22) at (-2,-2)  {$\bullet$};
\node (A23) at (-3,-3)  {$\bullet$};
\draw[very thick] (A2) to (A21b);
\draw[very thick] (A2) to (A22);
\draw[very thick] (A22) to (A22b);
\draw[very thick] (A22) to (A23);
\draw[very thick] (A23) to (A23b);
\draw[very thick] (A23) to (A24b);
\node (B1) at (3,1)  {$\bullet$};
\node (B2) at (3,-1)  {$\bullet$};
\node (B11) at (2,2)  {$\bullet$};
\node (B12) at (4,2)  {$\bullet$};
\node (B11b) at (1,3) [rectangle,rounded corners,draw=black] {$12$};
\node (B11c) at (3,3) [rectangle,rounded corners,draw=black] {$11$};
\node (B12b) at (5,3) [rectangle,rounded corners,draw=black] {$10$};
\node (B12c) at (5,1) [rectangle,rounded corners,draw=black] {$9$};
\draw[very thick] (B) to (B1);
\draw[very thick] (B) to (B2);
\draw[very thick] (B1) to (B11);
\draw[very thick] (B1) to (B12);
\draw[very thick] (B11) to (B11b);
\draw[very thick] (B11) to (B11c);
\draw[very thick] (B12) to (B12b);
\draw[very thick] (B12) to (B12c);
\node (B21) at (2,-2)  {$\bullet$};
\node (B22) at (4,-2)  {$\bullet$};
\node (B21b) at (1,-3) [rectangle,rounded corners,draw=black] {$16$};
\node (B21c) at (3,-3) [rectangle,rounded corners,draw=black] {$15$};
\node (B22b) at (5,-3) [rectangle,rounded corners,draw=black] {$14$};
\node (B22c) at (5,-1) [rectangle,rounded corners,draw=black] {$13$};
\draw[very thick] (B2) to (B21);
\draw[very thick] (B2) to (B22);
\draw[very thick] (B21) to (B21b);
\draw[very thick] (B21) to (B21c);
\draw[very thick] (B22) to (B22b);
\draw[very thick] (B22) to (B22c);

\end{tikzpicture}
\end{displaymath}
and 
\begin{displaymath}
  \tau'=\begin{tikzpicture}[xscale=1,yscale=1,baseline=0]
\node (A14b) at (-4,4) [rectangle,rounded corners,draw=black] {$4$};
\node (A13b) at (-2,4) [rectangle,rounded corners,draw=black] {$3$};
\node (A12b) at (-1,3) [rectangle,rounded corners,draw=black] {$2$};
\node (A1b) at (0,2) [rectangle,rounded corners,draw=black] {$1$};
\node (A) at (0,0)  {$\bullet$};
\node (A1) at (-1,1)  {$\bullet$};
\node (A2) at (3,-1)  {$\bullet$};
\node (A12) at (-2,2)  {$\bullet$};
\node (A13) at (-3,3)  {$\bullet$};
\node (A24b) at (6,-4) [rectangle,rounded corners,draw=black] {$8$};
\node (A23b) at (4,-4) [rectangle,rounded corners,draw=black] {$7$};
\node (A22b) at (3,-3) [rectangle,rounded corners,draw=black] {$6$};
\node (A21b) at (2,-2) [rectangle,rounded corners,draw=black] {$5$};
\node (A22) at (4,-2)  {$\bullet$};
\node (A23) at (5,-3)  {$\bullet$};
\node (B) at (2,0)  {$\bullet$};
\node (B1) at (3,1)  {$\bullet$};
\node (B2) at (-1,-1)  {$\bullet$};
\node (B11) at (2,2)  {$\bullet$};
\node (B12) at (4,2)  {$\bullet$};
\node (B11b) at (1,3) [rectangle,rounded corners,draw=black] {$12$};
\node (B11c) at (3,3) [rectangle,rounded corners,draw=black] {$11$};
\node (B12b) at (5,3) [rectangle,rounded corners,draw=black] {$10$};
\node (B12c) at (5,1) [rectangle,rounded corners,draw=black] {$9$};
\node (B21) at (0,-2)  {$\bullet$};
\node (B22) at (-2,-2)  {$\bullet$};
\node (B21b) at (1,-3) [rectangle,rounded corners,draw=black] {$16$};
\node (B21c) at (-1,-3) [rectangle,rounded corners,draw=black] {$15$};
\node (B22b) at (-3,-3) [rectangle,rounded corners,draw=black] {$14$};
\node (B22c) at (-3,-1) [rectangle,rounded corners,draw=black] {$13$};
\draw[very thick] (A) to (B);
\draw[very thick] (A) to (A1);
\draw[very thick] (A) to (B2);
\draw[very thick] (A1) to (A1b);
\draw[very thick] (A1) to (A12);
\draw[very thick] (A12) to (A12b);
\draw[very thick] (A12) to (A13);
\draw[very thick] (A13) to (A13b);
\draw[very thick] (A13) to (A14b);
\draw[very thick] (A2) to (A21b);
\draw[very thick] (A2) to (A22);
\draw[very thick] (A22) to (A22b);
\draw[very thick] (A22) to (A23);
\draw[very thick] (A23) to (A23b);
\draw[very thick] (A23) to (A24b);
\draw[very thick] (B) to (B1);
\draw[very thick] (B) to (A2);
\draw[very thick] (B1) to (B11);
\draw[very thick] (B1) to (B12);
\draw[very thick] (B11) to (B11b);
\draw[very thick] (B11) to (B11c);
\draw[very thick] (B12) to (B12b);
\draw[very thick] (B12) to (B12c);
\draw[very thick] (B2) to (B21);
\draw[very thick] (B2) to (B22);
\draw[very thick] (B21) to (B21b);
\draw[very thick] (B21) to (B21c);
\draw[very thick] (B22) to (B22b);
\draw[very thick] (B22) to (B22c);
\end{tikzpicture}
\end{displaymath}
which display both the split size sequence
\begin{displaymath}
s(\tau)=s(\tau')=(2,2,2,2,2,2,3,3,4,4,4,4,8). 
\end{displaymath}
So $\tau$ and $\tau'$ have the same split size sequence but differ in shape. By exhaustive search we found different (more complex) tree shapes with the same split size sequence already for  $n=11$. For  $n\le 10$, all binary trees with the same split size sequence have the same shape.
 \end{example}
 
So far we have analyzed the maximum of $\mathcal S_n$ and have seen that it is achieved uniquely by the caterpillar tree. However, we also want to find the most balanced tree shape(s), i.e. Pareto minima of $\mathcal S_n$. Surprisingly, they are not unique, as we will demonstrate now.
 
\begin{example}\label{ex:min8}
  Let $n=8$ and consider the  trees $\tau,\tau'\in \mathcal T_8^2$ with split size sequences 
  \begin{eqnarray*}
    s(\tau)&=&(2,2,2,2,4)\\
s(\tau')&=&(2,2,2,3,3)
  \end{eqnarray*}
which are depicted in Figure \ref{fig:2minima}. One could think of $\tau$ being mostly balanced with respect to the edge $1234|5678$ and $\tau'$ being mostly balanced with respect to the vertex $123|45|678$. 
For the monotonously increasing function $f$ with $$f(k)=\left\{
  \begin{array}[c]{cl}
    1&k\ge4\\
0&\mbox{otherwise}
  \end{array}
\right., $$ we derive $\Phi_{f}(\tau)=1>0=\Phi_f(\tau')$. On the other hand, for the monotonously increasing function $g$ with $$g(k)=\left\{
  \begin{array}[c]{cl}
    1&k\ge3\\
0&\mbox{otherwise}
  \end{array}
\right., $$ we derive $\Phi_{g}(\tau)=1<2=\Phi_g(\tau').$

\begin{figure}[htbp]
  \centering
  \begin{displaymath}
\tau=\qquad\begin{tikzpicture}[xscale=1,yscale=1, baseline =50 ]
\node (1) at (2,4) [rectangle,rounded corners,draw=black] {$1$};
\node (2) at (1,3) [rectangle,rounded corners,draw=black] {$2$};
\node (3) at (1,1) [rectangle,rounded corners,draw=black] {$3$};
\node (4) at (2,0) [rectangle,rounded corners,draw=black] {$4$};
\node (12) at (2,3)  {$\bullet$};
\node (34) at (2,1)  {$\bullet$};
\node (1234) at (3,2)  {$\bullet$};
\node (5678) at (4,2)  {$\bullet$};
\node (56) at (5,1)  {$\bullet$};
\node (78) at (5,3)  {$\bullet$};
\node (5) at (5,0) [rectangle,rounded corners,draw=black] {$5$};
\node (6) at (6,1) [rectangle,rounded corners,draw=black] {$6$};
\node (7) at (6,3) [rectangle,rounded corners,draw=black] {$7$};
\node (8) at (5,4) [rectangle,rounded corners,draw=black] {$8$};
\draw[very thick] (1) to (12);
\draw[very thick] (2) to (12);
\draw[very thick] (3) to (34);
\draw[very thick] (4) to (34);
\draw[very thick] (12) to (1234);
\draw[very thick] (34) to (1234);
\draw[very thick] (1234) to (5678);
\draw[very thick] (5678) to (56);
\draw[very thick] (5678) to (78);
\draw[very thick] (56) to (5);
\draw[very thick] (56) to (6);
\draw[very thick] (78) to (7);
\draw[very thick] (78) to (8);
\end{tikzpicture}
\end{displaymath}
\begin{displaymath}
\tau'=\qquad\begin{tikzpicture}[xscale=1,yscale=1, baseline =50 ]
\node (1) at (2,4) [rectangle,rounded corners,draw=black] {$1$};
\node (2) at (1,3) [rectangle,rounded corners,draw=black] {$2$};
\node (12) at (2,3)  {$\bullet$};
\node (3) at (2,1) [rectangle,rounded corners,draw=black] {$3$};
\node (123) at (3,2)  {$\bullet$};
\node (12345) at (4,2)  {$\bullet$};
\node (45) at (5,1)  {$\bullet$};
\node (678) at (5,3)  {$\bullet$};
\node (67) at (6,3)  {$\bullet$};
\node (4) at (5,0) [rectangle,rounded corners,draw=black] {$4$};
\node (5) at (6,1) [rectangle,rounded corners,draw=black] {$5$};
\node (8) at (5,4) [rectangle,rounded corners,draw=black] {$8$};
\node (6) at (7,2) [rectangle,rounded corners,draw=black] {$6$};
\node (7) at (7,4) [rectangle,rounded corners,draw=black] {$7$};
\draw[very thick] (1) to (12);
\draw[very thick] (2) to (12);
\draw[very thick] (3) to (123);
\draw[very thick] (12) to (123);
\draw[very thick] (123) to (12345);
\draw[very thick] (12345) to (678);
\draw[very thick] (12345) to (45);
\draw[very thick] (678) to (8);
\draw[very thick] (678) to (67);
\draw[very thick] (45) to (4);
\draw[very thick] (45) to (5);
\draw[very thick] (67) to (6);
\draw[very thick] (67) to (7);
\end{tikzpicture}
\end{displaymath}
\caption{The structure of the trees $\tau$ and $\tau'$ with different Pareto minima of $\mathcal S_8$.}
  \label{fig:2minima}
\end{figure}
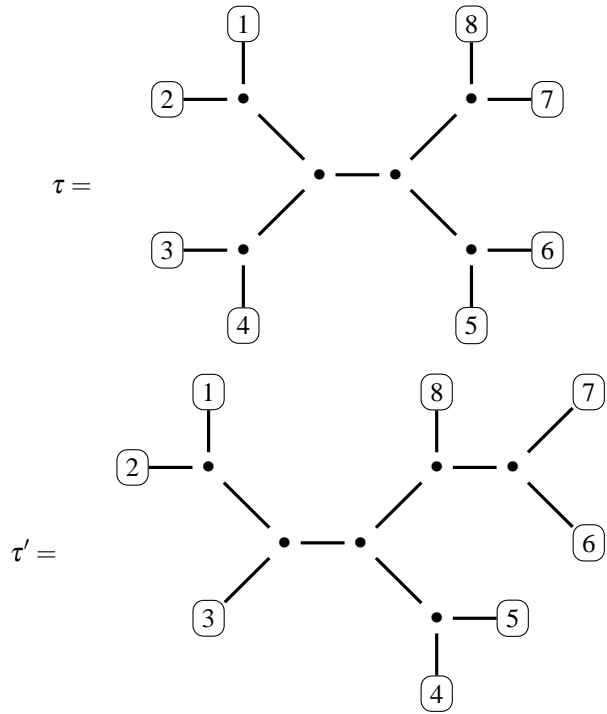

So depending on the underlying monotonously increasing function, either tree can be better than the other one. However, that the split size sequences induced by these trees are both Pareto minimal, i.e. that there is no tree which is more balanced than these two trees, can be seen when considering the hypothetical split size sequence $(2,2,2,2,3)$. This sequence would clearly dominate the above sequences, but $(2,2,2,2,3)\notin\mathcal{S}_8$: There are only eight leaves, so if there are four splits of size 2, this implies that all leaves should be in cherries. However, this contradicts the existence of a split of size 3.
\end{example}

As we have seen above, Pareto minima are not necessarily unique. The above example employed eight leaves. The following example shows that for $n=12$, there are even more Pareto minima, namely three.
 
\begin{example}
  There are more than 2 Pareto minima for large $n$. E.g., for $n=12$,
  \begin{equation}\label{eq:12splits}
    \begin{array}[c]{l}
      (2,    2,    2,    2,    2,    2,4,4,4)\\
   (2,    2,    2,    2,    2,3,3,4,5)\\
 (2,    2,    2,    2,3,3,3,3,6)
    \end{array}
  \end{equation}
are all minimal with respect to $\prec$.

We explicitly calculated the number of Pareto minima of $\mathcal{S}_n$ up to $n=22$. Surprisingly, the number of Pareto minima is not monotonous, see Figure \ref{fig:NoSn}. The sequence of the numbers of Pareto minima is not contained in the online encyclopedia of integer sequences \cite{OEIS}, which means that it seems to be unrelated to problems inducing other known integer sequences.  
\begin{figure}
  \centering
  
\includegraphics[angle=-90,width=0.7\textwidth]{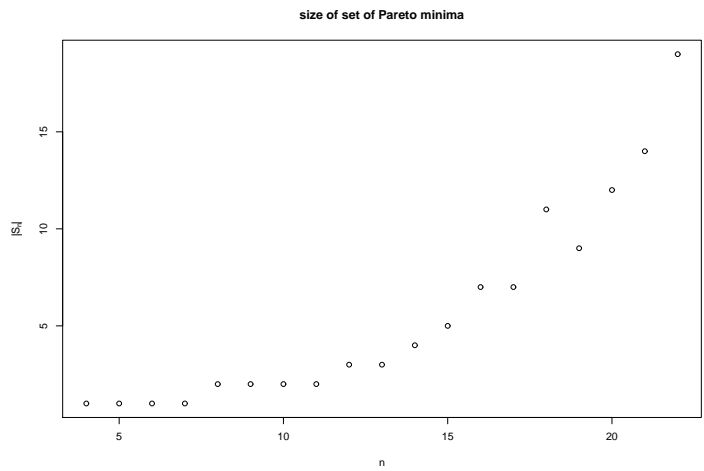}
  \caption{Number of Pareto minima of $\mathcal S_n$ for  $n=4,5,\dots,22$.}
  \label{fig:NoSn}
\end{figure}

\end{example}

Next we want to focus on NNI-moves and the neighborhoods they induce. 

\begin{definition}
  Let $f$ be increasing. Then $\tau\in\mathcal T_n$ is called an NNI-local minimum of $\Phi_f$ if for $\tau'$ in the 1-NNI-neighborhood of $\tau$, i.e. in the set of trees which can be reached from $\tau$ by performing one NNI-move, we have $\Phi_f(\tau')\ge\Phi_f(\tau)$.
\end{definition}
\begin{lemma}
  Let $\tau\in \mathcal T_n$ be a tree with its leaves labelled by $X$, $|X|=n$, and $f$ be strictly increasing. Then $\tau\in\mathcal T_n$ is an NNI-local minimum of $\Phi_f$ if and only if for all splits $\sigma$ such that $\sigma=X_1 \cup X_2 | X_3 \cup X_4$ with $X_i \cap X_j=\emptyset$ for all $i,j \in \{1,2,3,4\},$ $i \neq j$, and $\bigcup\limits_{i=1}^4 X_i = X$, \begin{displaymath}
\tau=  \begin{tikzpicture}[xscale=0.5,yscale=0.2, baseline =10 ]
\node (A) at (0,4) [rectangle,draw=black] {$X_1$};
\node (B) at (0,0) [rectangle,draw=black] {$X_2$};
\node (C) at (2,2)  {$\bullet$};
\node (D) at (4,2)  {$\bullet$};
\node (E) at (6,0) [rectangle,draw=black] {$X_3$};
\node (F) at (6,4) [rectangle,draw=black] {$X_4$};
\draw[very thick] (A) to (C);
\draw[very thick] (B) to (C);
\draw[very thick] (C) to (D);
\draw[very thick] (D) to (E);
\draw[very thick] (D) to (F);
\end{tikzpicture}
\end{displaymath}
we have
\begin{equation} 
  \min(x_1+x_2,x_3+x_4)\le \min \left\{
    \begin{array}[c]{c}
      \min(x_1+x_3,x_2+x_4),\\
\min(x_1+x_4,x_2+x_3)
    \end{array}
\right\}\label{eq:locNNI}
\end{equation}
\end{lemma}
\begin{proof}
   $\tau'$ in the 1-NNI-neighborhood of $\tau$ contains without loss of generality  the edge $X_1 \cup X_3 | X_2 \cup X_4$ as depicted below.
\begin{displaymath}
\tau'=  \begin{tikzpicture}[xscale=0.5,yscale=0.2, baseline =10 ]
\node (A) at (0,4) [rectangle,draw=black] {$X_1$};
\node (B) at (0,0) [rectangle,draw=black] {$X_3$};
\node (C) at (2,2)  {$\bullet$};
\node (D) at (4,2)  {$\bullet$};
\node (E) at (6,0) [rectangle,draw=black] {$X_2$};
\node (F) at (6,4) [rectangle,draw=black] {$X_4$};
\draw[very thick] (A) to (C);
\draw[very thick] (B) to (C);
\draw[very thick] (C) to (D);
\draw[very thick] (D) to (E);
\draw[very thick] (D) to (F);
\end{tikzpicture}
\end{displaymath}
Note that $\tau$ and $\tau'$ differ only in this edge, i.e. all other splits are the same in $\Sigma(\tau)$ and $\Sigma(\tau')$. Thus, in order to conclude that $\tau$ is an NNI-local minimum, we need
\begin{displaymath}
  f(\min(x_1+x_2,x_3+x_4))\le f(
      \min(x_1+x_3,x_2+x_4)).
\end{displaymath}
Since $f$ is strictly increasing, this follows directly from $\min(x_1+x_2,x_3+x_4)\le \min(x_1+x_3,x_2+x_4)$. Using an analogous argument for the other possible tree $\tau'' = X_1 \cup X_4 | X_2 \cup X_3$ leads to the desired result.
\end{proof}

\begin{remark}
Note that (\ref{eq:locNNI}) is just NNI-local minimality with respect to the function $f(k)=k$.  
\end{remark}

We now consider the role of NNI concerning the Pareto minima discussed earlier.

\begin{lemma}
  Every Pareto minimum is an NNI-local minimum.
\end{lemma}
\begin{proof}
  Any NNI-move changes at most  one split size. Thus, Inequality (\ref{eq:locNNI}) follows immediately from Pareto minimality. 
\end{proof}

The following lemma is a direct conclusion of the above findings.

\begin{lemma}
  Let $f$ be strictly monotonously increasing. Then for any minimal point $\tau\in \mathcal T_n$ of $\Phi_f$ the split size sequence $s(\tau)$ is a Pareto minimum of $\mathcal S_n$. 
\end{lemma}

According to  \cite{SWW11}, we call a tree $\tau$ semi-regular, if for all representations 
\begin{equation}
\tau=\qquad\begin{tikzpicture}[xscale=1,yscale=0.6, baseline =30 ]
\node (x1) at (0,3) [rectangle,draw=black] {$X_1$};
\node (x2) at (0,1) [rectangle,draw=black] {$X_2$};
\node (z1) at (1,2)  {$\bullet$};
\node (z2) at (2,2)  {$\bullet$};
\node (z3) at (3,2)  {$\bullet$};
\node (ppp) at (4,2)  {$\cdots$};
\node (zk) at (5,2)  {$\bullet$};
\node (zk+1) at (6,2)  {$\bullet$};
\node (x3) at (7,3) [rectangle,draw=black] {$X_3$};
\node (x4) at (7,1) [rectangle,draw=black] {$X_4$};
\node (y1) at (2,0) [rectangle,draw=black] {$Y_1$};
\node (y2) at (3,0) [rectangle,draw=black] {$Y_2$};
\node (yk) at (5,0) [rectangle,draw=black] {$Y_k$};
\draw[very thick] (x1) to (z1);
\draw[very thick] (x2) to (z1);
\draw[very thick] (y1) to (z2);
\draw[very thick] (y1) to (z2);
\draw[very thick] (y2) to (z3);
\draw[very thick] (yk) to (zk);
\draw[very thick] (z1) to (z2);
\draw[very thick] (z2) to (z3);
\draw[very thick] (z3) to (ppp);
\draw[very thick] (ppp) to (zk);
\draw[very thick] (zk) to (zk+1);
\draw[very thick] (x3) to (zk+1);
\draw[very thick] (x4) to (zk+1);
\end{tikzpicture}
\label{fig:semiregular}\end{equation}
with (without loss of generality) $x_1\le x_2$ and $x_3\le x_4$, we additionally have $x_2\le x_3$ or $x_4\le x_1$. There is a unique semi-regular tree shape for every $n\ge 4$ which is completely characterized. 

\begin{remark}
  Inequality (\ref{eq:locNNI}) is just semi-regularity of $\tau$ as defined above and in \cite{SWW11}, but restricted to adjacent vertices.
  \end{remark}

We are now in the position to state and prove the following theorem.

\begin{theorem}\label{thm:regular}
  Let $\map f{\set{1,\dots,\lfloor n/2\rfloor}}\NRp$ be strictly increasing and strictly concave, i.e. $2f(k)\ge f(k-1)+f(k+1)\mbox{ for all } k=2,\dots,\lfloor n/2\rfloor-1$. If  $\tau$ is a minimal point of $\Phi_f$ then $\tau$ is semi-regular.    
\end{theorem}
\begin{proof}
  First observe that we can extend $f$ to a concave function $\map f{\set{1,\dots,n-1}}\NRp$ with $f(n-k)=f(k)$. 
Now consider the situation in (\ref{fig:semiregular}) and assume $x_2>x_3$ and $x_4>x_1$. Fix the trees 
\begin{displaymath}
\tau_1=\qquad\begin{tikzpicture}[xscale=1,yscale=0.6, baseline =30 ]
\node (x1) at (0,3) [rectangle,draw=black] {$X_1$};
\node (x2) at (0,1) [rectangle,draw=black] {$X_3$};
\node (z1) at (1,2)  {$\bullet$};
\node (z2) at (2,2)  {$\bullet$};
\node (z3) at (3,2)  {$\bullet$};
\node (ppp) at (4,2)  {$\cdots$};
\node (zk) at (5,2)  {$\bullet$};
\node (zk+1) at (6,2)  {$\bullet$};
\node (x3) at (7,3) [rectangle,draw=black] {$X_2$};
\node (x4) at (7,1) [rectangle,draw=black] {$X_4$};
\node (y1) at (2,0) [rectangle,draw=black] {$Y_1$};
\node (y2) at (3,0) [rectangle,draw=black] {$Y_2$};
\node (yk) at (5,0) [rectangle,draw=black] {$Y_k$};
\draw[very thick] (x1) to (z1);
\draw[very thick] (x2) to (z1);
\draw[very thick] (y1) to (z2);
\draw[very thick] (y1) to (z2);
\draw[very thick] (y2) to (z3);
\draw[very thick] (yk) to (zk);
\draw[very thick] (z1) to (z2);
\draw[very thick] (z2) to (z3);
\draw[very thick] (z3) to (ppp);
\draw[very thick] (ppp) to (zk);
\draw[very thick] (zk) to (zk+1);
\draw[very thick] (x3) to (zk+1);
\draw[very thick] (x4) to (zk+1);
\end{tikzpicture}
\end{displaymath}
and
\begin{displaymath}
\tau_2=\qquad\begin{tikzpicture}[xscale=1,yscale=0.6, baseline =30 ]
\node (x1) at (0,3) [rectangle,draw=black] {$X_4$};
\node (x2) at (0,1) [rectangle,draw=black] {$X_2$};
\node (z1) at (1,2)  {$\bullet$};
\node (z2) at (2,2)  {$\bullet$};
\node (z3) at (3,2)  {$\bullet$};
\node (ppp) at (4,2)  {$\cdots$};
\node (zk) at (5,2)  {$\bullet$};
\node (zk+1) at (6,2)  {$\bullet$};
\node (x3) at (7,3) [rectangle,draw=black] {$X_3$};
\node (x4) at (7,1) [rectangle,draw=black] {$X_1$};
\node (y1) at (2,0) [rectangle,draw=black] {$Y_1$};
\node (y2) at (3,0) [rectangle,draw=black] {$Y_2$};
\node (yk) at (5,0) [rectangle,draw=black] {$Y_k$};
\draw[very thick] (x1) to (z1);
\draw[very thick] (x2) to (z1);
\draw[very thick] (y1) to (z2);
\draw[very thick] (y1) to (z2);
\draw[very thick] (y2) to (z3);
\draw[very thick] (yk) to (zk);
\draw[very thick] (z1) to (z2);
\draw[very thick] (z2) to (z3);
\draw[very thick] (z3) to (ppp);
\draw[very thick] (ppp) to (zk);
\draw[very thick] (zk) to (zk+1);
\draw[very thick] (x3) to (zk+1);
\draw[very thick] (x4) to (zk+1);
\end{tikzpicture}
\end{displaymath}

We set now  $\lambda=\frac{x_2-x_3}{x_2-x_3+x_4-x_1}$, and $z_j=\sum_{m=1}^jy_j$ where $z_0=0$. This way we obtain  from strict concavity and $$x_1+x_2+z_j=\lambda (x_1+x_3+z_j) +(1-\lambda)(x_4+x_2+z_j)$$that
\begin{eqnarray*}
 \lefteqn{ \Phi_f(\tau)-  \lambda \Phi_f(\tau_1)-(1-\lambda)\Phi_f(\tau_2)}\\
&=&\sum_{j=0}^k
f(x_1+x_2+z_j)-\lambda f(x_1+x_3+z_j) -(1-\lambda)f(x_4+x_2+z_j)\\
&>&0.
\end{eqnarray*}
This contradicts minimality of $\tau$, so the assumption was wrong. This leads to the desired result.
\end{proof}
\begin{remark}
  Interestingly, the trees $\tau_1,\tau_2$ are exactly at NNI distance 2 from $\tau$.  

It looks like strict concavity is not needed, but we are still lacking a proof for this case. Furthermore, we see that the NNI-neighborhood for two NNI-moves gives a sufficient criterion for minimality (and maximality) and that the NNI-neighborhood for one NNI-move is not sufficient for $n\ge8$.  This is consistent with the empirical observation that tree search algorithms using (too small) NNI-neighborhoods often get stuck in local minima \cite{KL15}.   
\end{remark}

 \begin{example} Also for non-concave $f$ the semi-regular pattern can be the minimal point. For $n=12$ we saw that there are 3 different Pareto minima. 
 For the convex function $f(k)=k^2$ we obtain the values $72,79,88$ for the split sequences in (\ref{eq:12splits}). This means the semi-regular pattern $(2,2,2,2,2,2,4,4,4)$ is still the minimizer.  
\end{example}



We conclude this section with the following lemma which is used in a later example.

 \begin{lemma}\label{bounds}
  Let $\tau \in \mathcal{T}_n^2$ with $n\geq 4$. Then, we have $\lceil  n/3\rceil \le s(\tau)_{n-3}\le\lfloor{ n}/2\rfloor.$
 \end{lemma}
 \begin{proof}
   Let $\sigma$ be the split the size of which is  $ s(\tau)_{n-3}$. Consider the edge $e$ of $\tau$ corresponding to $\sigma$. Note that $e$ splits $\tau$ into two parts, which we call the left-hand part and the right-hand part. Moreover, $e$  must be an inner edge (as $n \geq 4$ and thus there is at least one inner edge giving rise to one non-trivial split, so the split of size $ s(\tau)_{n-3}$, which is maximal, has to refer to a non-trivial split). Without loss of generality assume that the left-hand part of $\tau$ corresponds to $s(\tau)_{n-3}$. Then, the left-hand side end vertex of edge $e$ must have degree 3 (as $e$ is an inner edge). The splits corresponding to the  two other edges coincident with this vertex must have a size not larger than $ s(\tau)_{n-3}$. Thus $3 \cdot s(\tau)_{n-3}\ge n$ and the lower bound is derived. The upper bound is trivial.
 \end{proof}

  \section{Applications}

As we already mentioned in the introduction, functionals of the kind analyzed in the present manuscript have recently occurred in various contexts, some of which we want to mention here.

  \begin{example}\label{ex:gamma}
    In \cite{HW12} the authors considered the functional
    \begin{displaymath}
      \Gamma(\tau)=\sum_{A|B\in \Sigma^*(\tau)} |A|\cdot |B|,
    \end{displaymath}
    which is equivalent to $\Phi_f$ for $f(k)=k\cdot(n-k)$. They
    considered the functional only on $\mathcal T_n^2$, but for the
    maximum there is no difference anyway.  The application of  Theorem \ref{thm:1} is just the same as 
    \cite[Lemma 4.1]{HW12}, where we adapted our proof of Theorem \ref{thm:1} from. Further, the minimum is attained for the semi-regular shape, since $f$ is even strictly increasing.

    There is also a close relation between $\Gamma(\tau)$ and the Wiener index of $\tau$ as described in the next section, see also \cite{SWW11}.
  \end{example}

  \begin{example}\label{ex:MP} Recently, the maximum parsimony distance was defined independently in \cite{FK14} and \cite{MW15}. We will now briefly explain this concept before we show how it is related to the topic of this manuscript.  
  
Recall that a character $\chi$ is a function $\chi: X_\tau \rightarrow \mathcal{C}$ from the leaf set $X_\tau$ of $\tau$ to an alphabet  $\mathcal{C}$. In biology, $\mathcal{C}$ often refers to the four nucleotides in the DNA alphabet, but we do not restrict the definition to this case. Then, an extension $\bar{\chi}: V(\tau)   \rightarrow \mathcal{C}$ is a function from the vertex set $V(\tau)$ of $\tau$ to $\mathcal{C}$ which agrees with $\chi$ on the leaf set $X$. For a character $\chi$ with extension $\bar{\chi}$, the changing number of  $\bar{\chi}$, denoted $ch(\bar{\chi})$, is defined as the number of edges $e=\{u,v\} $ such that $\bar{\chi}(u)\neq \bar{\chi}(u)$, and the parsimony score $PS(\chi,\tau)$ is then defined as the minimum number of changes over all extensions $\bar{\chi}$ of $\chi$ on $\tau$, i.e. $PS(\chi,\tau)=\min\limits_{\bar{\chi}} ch(\bar{\chi})$. Note that the parsimony score of a character on a tree can be easily calculated with the famous Fitch algorithm \cite{fitch}. 

Now, the parsimony distance between two trees $\tau_1$ and $\tau_2$ is defined as follows: $d_{MP}(\tau_1,\tau_2)=\max_{\chi} |PS(\chi,\tau_1)-PS(\chi,\tau_2)|$. It has been shown \cite{FK14,FK15} that it is NP-hard to calculate this distance for two given trees (even if both trees are binary). We now consider neighborhoods induced by this metric.

In \cite{MW15}, the size of the 1-neighborhood of a tree $\tau\in \mathcal T_n$ with respect to $d_{MP}$ was derived to be
    \begin{displaymath}
      n_p(\tau)=4\sum_{A|B\in \Sigma^*(\tau)}|A|\cdot |B| -4(n-2)(n-3)+2|V_2(\tau)|+6|V_3(\tau)|
    \end{displaymath}
where $V_q(\tau)$, $q=1,2,3$ is the set of vertices of degree $q$ after removal of all pendant edges. 
Clearly, $V_1(\tau)$ are just all cherries. Unfortunately,  $|V_2(\tau)|$ is not of the form $\Phi_f$, but we can provide bounds on the minimal and maximal values, since we have
\begin{displaymath}
  4\Gamma(\tau)+2(n-2-|V_1(\tau)|)\le n_p(\tau)+4(n-2)(n-3)\le 4\Gamma(\tau)+6(n-2-|V_1(\tau)|).
\end{displaymath}
So we get \begin{displaymath}
  4\Gamma(\tau)-c |V_1(\tau)|=\Phi_f(\tau)
\end{displaymath}
for
\begin{displaymath}
  f(k)=\left\{
    \begin{array}[c]{cl}
      4k\cdot(n-k)&k>2,\\
8(n-2)-c&k=2.
    \end{array}
\right.
\end{displaymath}
Clearly, for $c<8(n-2)$, particularly for $c<32$, $f$ is strictly increasing and we find maximal values at caterpillars. 

Note that if you restrict the parsimony distance to binary characters, denoted $d_{MP}^2$, and define $f(\tau)=d_{MP}^2(\tau,\tau_0)$, where $\tau_0$ denotes the star-tree, i.e. the tree with no inner edges, then by Lemma \ref{bounds}, we have $f(\tau) \leq \lfloor {n}/{2}\rfloor-1$ and $f(\tau) \geq \lfloor {n}/{3}\rfloor-1$. It can easily be seen that $f$ is maximized by all trees containing an edge inducing a split $\sigma=A|B$ with $|A|=\lfloor {n}/{2}\rfloor$ and $|B|=\lceil {n}/{2}\rceil$. Caterpillar trees have this property,   but they are not the only ones. On the other hand, $f$ is minimized by all trees containing a node which is adjacent to three subtrees of size at least $\lfloor{n}/{3}\rfloor$. Therefore, the minimum is not unique, either.

  \end{example}

\begin{example}
Consider the functional $\tau\mapsto |V_1(\tau)|$, the number of cherries. Clearly, we have to choose   $f(k)=\left\{
  \begin{array}[c]{cl}
    1&k=2,\\
0&\mbox{otherwise.}
  \end{array}
\right.$
Since $f$ is decreasing and convex (on ${\set{2,\dots,\lfloor n/2\rfloor}}$), the minimum number 2 of cherries is attained by caterpillars only. The maximal number of cherries is yielded by the semi-regular shapes.       
  \end{example}

  \begin{example}
    Another  application are estimates of the diameter of $\mathcal
    T_n$ with respect to the Gromov-type $\ell^p$-distances on
    $\mathcal T_n$ introduced in \cite{Lie15} for $p=1,2$. We do not
    want to repeat the lengthy definition of these distances here. It
    is enough to note that Example 3 from that paper computed
    $D_p(\tau,\tau')$ if $\tau$ and $\tau'$ differ by one split, say
    $A|B$. Then
    \begin{displaymath}
      D_p(\tau,\tau')=\left\{
        \begin{array}[c]{cl}
          \min(|A|,|B|)& \mbox{for } p=1,\\
          \sqrt{|A|\cdot|B|/n}&\mbox{for }p=2.
        \end{array}\right.
    \end{displaymath}
    So we fix the functions $\tilde f_p$, $p=1,2$, $\tilde f_1(k)=k$,
    $\tilde f_2(k)=\sqrt{k\cdot(1-k/n)}$ which are both strictly increasing
    and concave on $\set{2,\dots,\lfloor n/2\rfloor}$. Let $\tau_0\in\mathcal
    T_n$ denote the star tree, i.e. $\Sigma^*(\tau_0)=\emptyset$. We
    see for all trees $\tau\in\mathcal T_n $:
    \begin{displaymath}
      D_p(\tau_0,\tau)\le \Phi_{f_p}(\tau).
    \end{displaymath}
    Theorem \ref{thm:1} gives us for even $n$ the estimates
    \begin{displaymath}
      \mathrm{diam}_{D_1}(\mathcal T_n)\le 2\sum_{k=2}^{n/2} 2k= n^2-2n-4
    \end{displaymath}
    and
    \begin{displaymath}
      \mathrm{diam}_{D_2}(\mathcal T_n)\le 2\sum_{k=2}^{n/2} 2\sqrt{k\cdot(1-k/n)}\sim 4n^{3/2}\int_0^{1/2}\sqrt{x(1-x)}\mathrm{d}x=\pi n^{3/2}
    \end{displaymath}
    on the diameter of $\mathcal T_n$.

    Unfortunately, these bounds are less tight that the ones derived
    in \cite[Lemma 9]{Lie15}. Nevertheless the same technique would  apply to any other tree distance if we had estimates for Robinson-Foulds moves \cite{RF81} in terms of $\norm\sigma$.

    At least, Theorem \ref{thm:1} supports now the conjecture that the maximal distance
   is attained between two caterpillars. But note that  we are just
    dealing with upper bounds on the diameter here.
  \end{example}
  \begin{example}\label{ex:max}
    Another simple functional is defined through
    \begin{displaymath}
      \Phi(\tau)=\max\set{\norm\sigma:\sigma \in \Sigma^*(\tau)}=s(\tau)_{n-3}.
    \end{displaymath}
This functional is not of the form $\Phi_f$. Still, it is monotonously increasing with respect to $\prec$. Its extremal values $\lfloor n/2\rfloor,\lceil n/3\rceil$ were derived in Lemma \ref{bounds}. By the previous section, they are achieved  by caterpillars and some Pareto minimum, possibly among others. 

For $n=8$  we computed the Pareto minima, see Example \ref{ex:min8}. $(2,2,2,3,3)$ realizes the minimal value $3$ for $\Phi$. But, the other (semi-regular) split size sequence $(2,2,2,2,4)$ realizes the \emph{maximal} value for $\Phi$. 
  \end{example}
  
\section{Discussion}

We derived necessary and sufficient criteria to compute minimal and maximal points of the functionals $\Phi_f$. There were a lot of specific functionals of this kind considered in the past, see the previous section. By our theory, it is now less surprising that quite often caterpillars yielded maximal values and the semi-regular trees yielded minimal ones. 
But, we also saw that sometimes the minimum of $\Phi_f$ could be achieved by a different Pareto minimum. So the set of split size sequences $\mathcal S_n$ and their Pareto minima seems quite interesting to study.  For instance, it would be nice to know whether all Pareto minima, not only the semi-regular one, correspond to a unique tree shape.  

The $\Gamma-$index from Example \ref{ex:gamma} is closely related to the Wiener index of graphs, it is its leaf-restricted form. Essentially it holds
\begin{displaymath}
  \Gamma( \tau)=\sum_{u,v\in L(\tau)} d_\tau(u,v),
\end{displaymath}
where $d_\tau$ is the metric induced by $\tau$ and $L(\tau)$ are the leaves of $\tau$, see \cite{SWW11}.  This kind of index is natural, unique  and of course a shape invariant of $\tau$. 

It is easy to see that  we can derive a whole family of similar topological indices if we introduce another metric on the tree, still just depending on the shape of $\tau$. More precisely, introduce   
for any $\sigma\in\Sigma(\tau)$ a weight $w(\sigma,\tau)=g(\norm\sigma)\ge 0$ and the corresponding (semi-)metric $d_\tau^w$ on $L(\tau)$. Then, setting  
\begin{displaymath}
  \Gamma^w( \tau)=\sum_{u,v\in L(\tau)} d^w_T(u,v)
\end{displaymath}
we obtain $\Gamma^w(\tau)=\Phi_f(\tau)$ for $f(k)=g(k)k(n-k)$.  If $g$ is increasing, $f$ is increasing as well. Unfortunately, concavity is not so easy to derive. Nevertheless, our theory will apply to many of these $\Gamma^w-$indices.    
 
Surely, these indices are linear in $d_\tau$. It is easy to derive similar functionals which are quadratic in $d_\tau$ or depending on triple partitions $A|B|C$ compatible with the tree $\tau$. For instance, the functional $\Phi(\tau)=|V_2(\tau)|$ from Example \ref{ex:MP}  is of this kind. It would be quite valuable  to extend our theory to this type of functionals.


\begin{thebibliography}{99}

\bibitem{HW13} P.J. Humphreys and T. Wu, On the Neighborhoods of Trees, IEEE Trans. Comp. Biol. Bioinf. \textbf{10}(3):721--728, 2013. 
\bibitem{HW12} P.J. Humphreys and T. Wu, On the neighbourhoods of trees. preprint 2012 \texttt{arXiv:1202.2203} (\texttt{Arxiv} version of \cite{HW13})
\bibitem{KL15} S. Kl\"are and J. Leigh, personal communication
\bibitem{Lie15} V. Liebscher, Gromov meets Phylogenetics --- new Animals for the Zoo of Metrics on Tree Space. submitted 2015 \texttt{arXiv:1504.05795}
\bibitem{SWW11} L.A. Sz\'ekely, H. Wang, and  T. Wu, The sum of the distances between the leaves of a
tree and the 'semi-regular' property. Discr. Math.
\textbf{311}(13):1197--1203, 2011.
\bibitem{SW05}L.A. Sz\'ekely and H. Wang, On subtrees of trees, Adv. Appl. Math.\textbf{34}(1): 138--155, 2005
\bibitem{MW15} V. Moulton and T. Wu, A parsimony-based metric for phylogenetic trees, Adv. Appl. Math. \textbf{66}: 22--45, 2015.
\bibitem{newick} J. Felsenstein, J. Archie, W. Day, W. Maddison, C. Meacham, F. Rohlf and D. Swofford, The newick tree format. http://evolution.genetics.washington.edu/phylip/newicktree.html, 2000.
\bibitem{FK14} M. Fischer and S. Kelk, On the Maximum Parsimony distance between phylogenetic trees, in press at Annals of Combinatorics.
\bibitem{FK15} S. Kelk and M. Fischer,  On the complexity of computing MP distance between binary phylogenetic trees, submitted to Annals of Combinatorics. 
\bibitem{fitch} W.M. Fitch, Toward defining the course of evolution: minimum change for a specific tree topology, Syst. Zool.\textbf{20}(4): 406--416, 1971.
\bibitem{RF81} D.R. Robinson, L.R. Foulds, Comparison of phylogenetic trees, Math. Biosciences \textbf{53}: 131--147, 1981. 
\bibitem{OEIS}The On-Line Encyclopedia of Integer Sequences, published electronically at \texttt{https://oeis.org/}, 2015
\end{thebibliography}
\end{document}